\documentclass{article} % For LaTeX2e
\usepackage{iclr_aims,times, amsthm}

\usepackage{amsmath,amsfonts,bm}

\def\eqref#1{equation~\ref{#1}}
\def\1{\bm{1}}

\DeclareMathAlphabet{\mathsfit}{\encodingdefault}{\sfdefault}{m}{sl}
\SetMathAlphabet{\mathsfit}{bold}{\encodingdefault}{\sfdefault}{bx}{n}

\DeclareMathOperator*{\argmax}{arg\,max}

\usepackage{hyperref}
\usepackage{url}

\newtheorem{proposition}{Proposition}
\newtheorem{lemma}{Lemma}
\newtheorem{corollary}{Corollary}

\theoremstyle{definition}

\theoremstyle{remark}

\newcommand{\mexpect}[2]{\mathbb{E}_{#1}\left[ #2\right]}

\newcommand{\prob}[1]{\mathbb{P}\left( #1 \right)}

\newcommand{\oo}[1]{\left( #1 \right)}

\newcommand{\co}[1]{\left[ #1 \right)}
\newcommand{\cc}[1]{\left[ #1 \right]}

\title{Very Credible Auction}

\iclrfinalcopy

\author{Thanawat Sornwanee \\
Graduate School of Business\\
Stanford University\\
Stanford, CA 94305, USA \\
\texttt{tsornwanee@stanford.edu}}

\begin{document}

\maketitle

\begin{abstract}
    We introduce a notion and motivation for very credible auctions. This is a strengthening in the commitment in the credible auctions setting, making it more suitable for a modern digital applications.
\end{abstract}

These days, people, effectively but unknowingly, participate in auctions\footnote{Pricing is also included as an auction with a single bidder.}, performing the auction bidding through the interaction with digital platforms. Digital platforms collect data from users, and use it for allocating its (possibly scarce) resources. At the time of the interactions, it is difficult to know or believe how such interactions or reports, which can be done in an informal and unstructured manner like through natural language inputting, will be used, or even whether it will be used. The time such data is used can also be distant from the time it is collected. The location as well as the agency can also be different, The data may be collected by a company, get processed and sold to another company, and eventually get used by yet another company. The algorithm of the company may also be handled by another company via a complicated and obscure revenue sharing system.

One can imagine that, in such case, the bounded rationality would take effect, making users truthfully answer almost every question. However, as a system designer, one may holds a desideratum of some level of \emph{strategic fairness}, ensuring that a cognitively, morally, or computationally constrained individual will not be too much at a disadvantage. A strong notion for this desideratum can be a prevalent notion of strategy-proofness and its stronger variants, such as group strategy-proofness and obvious strategy-proofness.~\citep{li2017obviously} However, strategy-proofness is hard to obtain, we will then only look for truthfulness as one of the equilibria. 

The interaction system is also somewhat different from a canonical auction model or a canonical mechanism design, where the auction rule is specified and known to the bidder prior to the action of bidding. In this informal auction, the rule is unknown, and, even if known, cannot be committed or enforced. The algorithm for a company may be updated before the allocation is computed. Therefore, one can think of the action taken by a company is done in an ad-hoc manner. This suggests that we need a new framework to study this new auction system, which requires less credibility than what would have previously needed.

In the section~\ref{section:very}, we introduce and define \emph{Very Credible Auction} through an equilibrium of a novel auction game.

\section{Very Credible Auctions}
\label{section:very}

Unlike the notion of credible auction~\citep{akbarpour2020credible}, as well as other variants~\citep{haupt2021contextually, komo2024shill}, the notion of very credible auction will be defined through an equilibrium of a newly introduced auction model, which is a signaling game with a subsequent sale offer. We will later show that such equilibrium of this abstract game can be converted into a very credible auction. This is done by choosing an appropriate question answering scheme so that truthful answering is part of the equilibrium strategy without changing the economic behavior of the equilibrium found via the abstract game.

We first informally describe the auction/interaction scheme. There exists a single product, which is valued by many bidders (online users/customers). Each bidder gets to privately communicate with the auctioneer (digital platform and pricing/allocating platform). This communication can be informal as well as uninformative, or it can also be formal via some feature experiment or product and preference survey. The auctioneer is also allowed to response back to the bidders, but also in private. For example, there can be a follow-up question or survey. The auctioneer may also go idle, and come back to interact with the bidders later. We allow this possibility as long as the sequential nature of the interaction is finalized before the allocation/pricing decision done by the auctioneer. Finally, the auctioneer chooses a single bidder to sell the product to. The price is chosen by the auctioneer, but the bidder can still reject the offer.

\subsection{Model}

There exist $n$ bidders and a single product in the system. 

Each bidder realizes the value of the product to be $V_i \sim \text{Unif}([0,1])$. The product is currently endowed to the auctioneer, whose intrinsic value for the product is $0$.

At time $t=1$, each bidder, in private, submits a bid, which can be a private communication\footnote{In a conventional auction, the bid is usually a price a bidder will buy from the auctioneer if win (as in first price auction) or other quantities related to price. However, in this setting, we allow the flexibility so it can capture other mode of communication including informationally irrelevant communication.}, to the auctioneer. The bid is a message $m \in \mathcal{M}$, which is the abstract collection of messages assumed to be sufficiently rich and regular, and we will from now on use the term \emph{message} exchangeably with \emph{bid}. We allow independent randomization, and denote the possibly random message sent by the bidder $i \in \{1,2,\dots,n\}$ as $M_i$.

At time $t=2$, upon receiving the total message
\begin{align*}
    M := \cc{M_i}_{i=1}^n,
\end{align*}
the auctioneer selects the bidder to whom such auctioneer will send a sale proposal as well as the ask price. Thus, the auctioneer chooses a possibly random $(I, P)$, where $I$ is the possibly random index of the winner from $\{1,2,\dots, N\}$, while $P$ is the possibly random price taking a value of $[0,1]$. 

At time $t=3$, the proposed bidder $I$ decides whether to accept the sale proposal and buy the product with price $P$ or to reject the proposal and has no transaction. The sale proposal is a take-it-or-leave-it offer\footnote{which is common in mechanism design literature}, and we assume that, once a rejection happens, the auctioneer cannot propose another sale proposal to the same winner nor to other bidders\footnote{This restriction prevents the sequential bargaining and price experimentation, but one can also consider these more dynamic framework as a possibly more realistic assumption for some auctions.}. We denote the possibly random indicator of whether the proposal is accepted as $B$.

The utility of the bidder $i \in \{1,2,\dots,n\}$ is
\begin{align*}
    U_i = 
    \begin{cases}
        B(V_i -P) &\text{ if } i = I\\
        0 &\text{ if } i \ne I
    \end{cases},
\end{align*}
meaning that the utility is zero if the transaction does not happen and is equal to the difference between the value and the price if the transaction is realized.

The utility for the auctioneer is the revenue
\begin{align*}
    R_i = BP,
\end{align*}
meaning that the utility is the price if the transaction occurs, and is otherwise $0$.

We are interested in a symmetric equilibrium consisting of perfect Bayesian equilibrium for the signaling part (at time $t=1$ and $t=2$) as well as a subgame perfect equilibrium for the response part (at time $t=3$).

\paragraph{Rationality to Accept/Reject}

At time $t=3$, it is obvious that, if $V_I > P$, the the bidder $I$ will accept the proposal, making $B=1$ with probability $1$. Similarly, if $V_I< P$, rationality also enforces that $B=0$ with probability $1$. In general, the behavior is not pinned down when $V_I$ is exactly $P$. However, we assume that 
\begin{align*}
    B = \mathbf{1}_{V_I \ge P},
\end{align*}
so tie-breaking of whether to accept the offer will be in favor of the auctioneer\footnote{This is rather standard, and also provides closeness, which can be important to achieve the existence of optimal price.~\citep{sornwanee20251}}. This uniqueness in the response will also help simplifying the equilibrium notion.

\subsection{Truthfulness}
Currently, the notion of truthfulness is meaningless, since there is no semantic constraint in the sending of the message $M_i$ for a bidder $i$. To make truthfulness matters, we introduce a simple question-answering scheme for the bidding/communication scheme. It is easy to show that an equilibrium can be converted into a truthful equilibrium by a direct revelation principle like argument.

\subsection{Definition}
A mechanism is specified by $\oo{\hat{M}, \kappa}$. The kernel $\hat{M}$ specifies the selection of message $M_i$ conditioned on the realized value $V_i$, and such kernel is shared across every bidder $i \in \{1,2,\dots, n\}$. The kernel $\kappa$ is a kernel mapping the realized full message $M \in \mathcal{M}^n$ to a sale proposal $(I,P)$.

As we have established that truthfulness can be done by choosing a question-answering complying with any arbitrary semantic restriction, we say that $\oo{\hat{M}, \kappa}$ is a very credible mechanism if $\oo{\hat{M}, \kappa}$ is a symmetric perfect Bayesian equilibrium for the abstract game specified previously.

\subsection{Connection to Credible Auction}

    Credible auction~\cite{akbarpour2020credible} is a deviation from a more conventional auction in the following manner: In a conventional auction, the auctioneer creates a rule of the game, specifying what the bidders can bid, and the subsequent event of the allocation and transaction conditioned on the realized joint bids. Such subsequent event often governs whether the auctioneer gets to keep the item, which bidders get the item (and in some case how much), and the price each bidder has to pay. This relies on a trustworthy third party to uphold the rule of the same. Credible auction considers the scenario when the bidding is private, making the realized bid known to only the corresponding bidder and the auctioneer. Therefore, the bidder cannot always ensure that the auctioneer commits to the rule of the game. For the mechanism that the bidders can distributedly ensure, we consider it a credible auction. 

    However, we can also view it as \emph{credibility} through the \emph{commitment} of the bidder to the rule of the game. For example, in first price auction without a reserve price as well as generalized first price auction~\citep{mehta2007adwords,ostrovsky2022pure}\footnote{which is when the auctioneer supposed to treat the report value as if it was a truth or as if it was from a VCG mechanism~\citep{groves1973incentives}.}, as long as the bidder can commit to not paying more than what they have bid, the auctioneer cannot profitably deviate due to monotonicity.

    We are seeking to relax such commitment. For example, in a first price auction setting, even when the winner has bid a price of $m$ but the actual valuation is $v>m$, the winner will still accept the offer made by the auctioneer at the price of $p\in (m,v)$. This is because the announcement to only accept the price less than or equal to $m$ is not credible in the first place. The auctioneer may come to offer the sale proposal via alias or through delegate, or the time of sale proposal may be already far fro the tie the bidding activity is finalized. Thus, the first price auction with symmetric and screening equilibrium, though credible, is not very credible.

\section{Single Agent Case}

    We have suggested that a first price auction does not satisfy the criteria for being a very credible auction in general. However, this is only true in a multi-bidder case where a symmetric equilibrium is to bid $M_i = \frac{n-1}{n} V_i$. Therefore, whenever a bidder $i$ wins the auction, the auctioneer under equilibrium can inversely calculate $V_i = \frac{n}{n-1} M_i$, and can select a price $P = \frac{n}{n-1} M_i$ for the bidder to face. Therefore, even if a bidder wins the auction, the surplus will still be $0$. 
    
    However, in a single agent case, first price auction with a reserve price is a very credible mechanism. This is because the bidding will be informationally irrelevant, allowing the auctioneer to be committed to the reserve price, which is also the monopoly price of $\frac{1}{2}$.

    Moreover, we can see that this scheme is the only very credible auction in this scenario.

    \begin{proposition}
        If the number of agents $n=1$, and the message space $\mathcal{M}$ is countable, then the sale proposal $(I,P) = \oo{1,\frac{1}{2}}$ almost surely under equilibrium.
    \end{proposition}
    This proposition suggests that the sale offer will not be altered by the bid $M$ received.
    If the pricing is deterministic to the message, then it is obvious that heterogeneity cannot be supported by an equilibrium. Otherwise, the bidder will always submit the bid that give the minimal price. Although this may suggest that randomization of price can sustain an equilibrium, we will see that this is not the case due to the fact that the auctioneer has to set a randomization over a set of monopoly prices with respect to the inferred posterior of $V_1$ conditioned on $M_1$.

    Under an equilibrium, for each message $m \in \mathcal{M}$, the platform/monopoly will uncommittedly set a randomized price
    \begin{align*}
        (P\vert M_1=m) \sim \mu_m \in \Delta([0,1]).
    \end{align*}

    The bidder then find a conditional utility for sending a message $m \in \mathcal{M}$ and having a type $v \in [0,1]$ as well as conditioned on the subsequent action of being sequential rational after the price is proposed to be
    \begin{align*}
        u^{(m)}(v) := \mexpect{P \sim \mu_m}{\max\oo{\left\{v-P,0\right\}}}.
    \end{align*}

    It is easy to see the following characterization.
    \begin{corollary}
    \label{corollary:convexcharcaterization}
        For any convex, continuous, non-decreasing function $u: [0,1] \to [0, \infty)$ such that $u(0) = 0$ and $\sup_{x \in [0, 1)} \sup(\partial u(x)) \le 1$, there exists a unique distribution $\mu \in \Delta([0,1])$ such that
        \begin{align*}
            u(x) = \mexpect{P \sim \mu}{\max\oo{\left\{x-P,0\right\}}}
        \end{align*}
        for all $x \in [0,1]$.

        For any distribution $\mu \in \Delta([0,1])$, the function $: [0,1] \to [0, \infty)$ defined such that
        \begin{align*}
            u(x) := \mexpect{P \sim \mu}{\max\oo{\left\{x-P,0\right\}}}
        \end{align*}
        for all $x \in [0,1]$ will be convex, continuous, and non-decreasing, have $u(0)=0$, and have $\sup_{x \in [0, 1)} \sup(\partial u(x)) \le 1$.
    \end{corollary}

    For notational simplicity, we extend the domain so that $u^{(m)}: [0,1.1] \to [-0.1,1.1]$, so we will have that the characterization is changed from $\sup_{x \in [0, 1)} \sup(\partial u(x)) \le 1$.

    \begin{lemma}[Conditional Monopoly Pricing]
    \label{lemma:block}
        
        For any $m \in \mathcal{M}$ with $\prob{M_1=m} > 0$, for any $x,y \in [0,1]$ with $y\ge x$, we have that
        \begin{align*}
            \sup \partial u^{(m)}(y) - \inf \partial u^{(m)}(x) > 0
        \end{align*}
        only if, for any $\epsilon > 0$, the measure $\Pi\oo{\co{x,y+\epsilon} \cap \left\{v \in [0,1]: m \in\argmax_{m' \in \mathcal{M}}u^{(m')}(v)\right\}} > 0$.
    \end{lemma}
    \begin{proof}
        We note that $\sup \partial u^{(m)}(y) - \inf \partial u^{(m)}(x) > 0$ if and only if there exists some monopoly price $p^{(m)}$ with respect to the posterior $\mu_m$ such that $p^{(m)} \in [x,y]$. Note that $\mu_m$ is upperbounded (for every measurable set) by $\frac{1}{\prob{M=m}} \Pi_m$ where $\Pi_m$ is the restriction of the probability $\Pi$ to the region of $v \in [0,1]$ where $m \in\argmax_{m' \in \mathcal{M}}u^{(m')}(v)$. Thus, if there exists some $\epsilon > 0$ such that
        \begin{align*}
            \Pi\oo{\co{x,y+\epsilon} \cap \left\{v \in [0,1]: m \in\argmax_{m' \in \mathcal{M}}u^{(m')}(v)\right\}} = 0,
        \end{align*}
        we can then raise the price from $p^{(m)}$ to $p^{(m)}+\frac{\epsilon}{2}$ and get a higher expected revenue, since the revenue will be $\frac{p^{(m)}+\frac{\epsilon}{2}}{p^{(m)}} > 1$ of the original revenue, which is non-zero since $\mu_m$ is not degenerate at $0$.
    \end{proof}

    From this lemma, we will see that, even when two distinct convex functions satisfying the conditions in the corollary~\ref{corollary:convexcharcaterization} can intersect many times, the monopoly with the lack of commitment power cannot commit to such convex functions. After the first convex function takes value below\footnote{without loss of generality} the second, the first convex function will not be able to attract such values of user to choose it. Thus, the function cannot change its slope and will forever be lower than the second function.

    \begin{lemma}
    \label{lemma:samefucntion}
        For any $m, m' \in \mathcal{M}$ with $\prob{M_1=m},\prob{M_1=m'} > 0$, then
        \begin{align*}
            u^{(m)}(v) = u^{(m')}(v)
        \end{align*}
        for all $v \in [0,1]$.
    \end{lemma}
    \begin{proof}
        If not, then, without loss of generality, there exists some $v \in [0,1]$ such that $u^{(m)}(v) < u^{(m')}(v)$. From the corollary~\ref{corollary:convexcharcaterization}, the discrepancy function $u^{(m)} - u^{(m')}$ is continuous. Thus, there exists a non-empty open interval $(a,b) \subseteq (0,1)$ such that $u^{(m)}(v) < u^{(m')}(v)$ for all $v \in (a,b)$, and $u^{(m)}(a) = u^{(m')}(a)$.
        
        Thus, 
        \begin{align*}
            \sup \partial u^{(m)}(a) < \sup \partial u^{(m')}(a),
        \end{align*}
        and the set
        \begin{align*}
            A := \left\{v \in [0,1]: m \in\argmax_{m' \in \mathcal{M}}u^{(m')}(v)\right\} \subseteq [0,1] - (a,b).
        \end{align*}

        From the lemma~\ref{lemma:block}, for any $v,v' \in (a,b)$, we will have that
        \begin{align*}
            \sup \partial u^{(m)}(v) = \inf \partial u^{(m)}(v'),
        \end{align*}
        so
        \begin{align*}
            \inf \partial u^{(m)}(b) = \sup_{v \in (a,b)} \sup \partial u^{(m)}(v) =
            \inf_{v' \in (a,b)} \inf \partial u^{(m)}(v') = \sup \partial u^{(m)}(a),
        \end{align*}
        while $\inf \partial u^{(m')}(v) \ge \sup \partial u^{(m')}(a)> \sup \partial u^{(m)}(a)$ for all $v > a$. Thus, we have that
        \begin{align*}
            u^{(m)}(b) - u^{(m')}(b) < 0.
        \end{align*}
        Therefore, there cannot be any $b \in (a, 1.1]$ such that $u^{(m)}(b) > u^{(m')}(b)$, making $\sup \partial u^{(m)}(1) < \sup \partial u^{(m)}(a) \le 1$, creating a contradiction with the corollary~\ref{corollary:convexcharcaterization}.    
    \end{proof}

    We can now proceed to prove the main proposition.
    \begin{proof}
        The first part of the main theorem follows directly from the lemma~\ref{lemma:samefucntion} and the corollary~\ref{corollary:convexcharcaterization}.

        The second part is obvious, since we specify the off-equilibrium path to be non-profitable.
    \end{proof}

\bibliography{egbib_full}
\bibliographystyle{iclr2026_conference}

\end{document}